\documentclass{article}

\usepackage[a4paper,total={5in,8.3in}]{geometry}

\usepackage{amsmath,amsfonts,amssymb,amsthm}
\usepackage{tikz}
\usetikzlibrary{calc,arrows}

\usepackage{hyperref}
\makeatletter
\def\PYG@reset{\let\PYG@it=\relax \let\PYG@bf=\relax%
    \let\PYG@ul=\relax \let\PYG@tc=\relax%
    \let\PYG@bc=\relax \let\PYG@ff=\relax}
\def\PYG@tok#1{\csname PYG@tok@#1\endcsname}
\def\PYG@toks#1+{\ifx\relax#1\empty\else%
    \PYG@tok{#1}\expandafter\PYG@toks\fi}
\def\PYG@do#1{\PYG@bc{\PYG@tc{\PYG@ul{%
    \PYG@it{\PYG@bf{\PYG@ff{#1}}}}}}}
\def\PYG#1#2{\PYG@reset\PYG@toks#1+\relax+\PYG@do{#2}}

\@namedef{PYG@tok@w}{\def\PYG@tc##1{\textcolor[rgb]{0.73,0.73,0.73}{##1}}}
\@namedef{PYG@tok@c}{\let\PYG@it=\textit\def\PYG@tc##1{\textcolor[rgb]{0.24,0.48,0.48}{##1}}}
\@namedef{PYG@tok@cp}{\def\PYG@tc##1{\textcolor[rgb]{0.61,0.40,0.00}{##1}}}
\@namedef{PYG@tok@k}{\let\PYG@bf=\textbf\def\PYG@tc##1{\textcolor[rgb]{0.00,0.50,0.00}{##1}}}
\@namedef{PYG@tok@kp}{\def\PYG@tc##1{\textcolor[rgb]{0.00,0.50,0.00}{##1}}}
\@namedef{PYG@tok@kt}{\def\PYG@tc##1{\textcolor[rgb]{0.69,0.00,0.25}{##1}}}
\@namedef{PYG@tok@o}{\def\PYG@tc##1{\textcolor[rgb]{0.40,0.40,0.40}{##1}}}
\@namedef{PYG@tok@ow}{\let\PYG@bf=\textbf\def\PYG@tc##1{\textcolor[rgb]{0.67,0.13,1.00}{##1}}}
\@namedef{PYG@tok@nb}{\def\PYG@tc##1{\textcolor[rgb]{0.00,0.50,0.00}{##1}}}
\@namedef{PYG@tok@nf}{\def\PYG@tc##1{\textcolor[rgb]{0.00,0.00,1.00}{##1}}}
\@namedef{PYG@tok@nc}{\let\PYG@bf=\textbf\def\PYG@tc##1{\textcolor[rgb]{0.00,0.00,1.00}{##1}}}
\@namedef{PYG@tok@nn}{\let\PYG@bf=\textbf\def\PYG@tc##1{\textcolor[rgb]{0.00,0.00,1.00}{##1}}}
\@namedef{PYG@tok@ne}{\let\PYG@bf=\textbf\def\PYG@tc##1{\textcolor[rgb]{0.80,0.25,0.22}{##1}}}
\@namedef{PYG@tok@nv}{\def\PYG@tc##1{\textcolor[rgb]{0.10,0.09,0.49}{##1}}}
\@namedef{PYG@tok@no}{\def\PYG@tc##1{\textcolor[rgb]{0.53,0.00,0.00}{##1}}}
\@namedef{PYG@tok@nl}{\def\PYG@tc##1{\textcolor[rgb]{0.46,0.46,0.00}{##1}}}
\@namedef{PYG@tok@ni}{\let\PYG@bf=\textbf\def\PYG@tc##1{\textcolor[rgb]{0.44,0.44,0.44}{##1}}}
\@namedef{PYG@tok@na}{\def\PYG@tc##1{\textcolor[rgb]{0.41,0.47,0.13}{##1}}}
\@namedef{PYG@tok@nt}{\let\PYG@bf=\textbf\def\PYG@tc##1{\textcolor[rgb]{0.00,0.50,0.00}{##1}}}
\@namedef{PYG@tok@nd}{\def\PYG@tc##1{\textcolor[rgb]{0.67,0.13,1.00}{##1}}}
\@namedef{PYG@tok@s}{\def\PYG@tc##1{\textcolor[rgb]{0.73,0.13,0.13}{##1}}}
\@namedef{PYG@tok@sd}{\let\PYG@it=\textit\def\PYG@tc##1{\textcolor[rgb]{0.73,0.13,0.13}{##1}}}
\@namedef{PYG@tok@si}{\let\PYG@bf=\textbf\def\PYG@tc##1{\textcolor[rgb]{0.64,0.35,0.47}{##1}}}
\@namedef{PYG@tok@se}{\let\PYG@bf=\textbf\def\PYG@tc##1{\textcolor[rgb]{0.67,0.36,0.12}{##1}}}
\@namedef{PYG@tok@sr}{\def\PYG@tc##1{\textcolor[rgb]{0.64,0.35,0.47}{##1}}}
\@namedef{PYG@tok@ss}{\def\PYG@tc##1{\textcolor[rgb]{0.10,0.09,0.49}{##1}}}
\@namedef{PYG@tok@sx}{\def\PYG@tc##1{\textcolor[rgb]{0.00,0.50,0.00}{##1}}}
\@namedef{PYG@tok@m}{\def\PYG@tc##1{\textcolor[rgb]{0.40,0.40,0.40}{##1}}}
\@namedef{PYG@tok@gh}{\let\PYG@bf=\textbf\def\PYG@tc##1{\textcolor[rgb]{0.00,0.00,0.50}{##1}}}
\@namedef{PYG@tok@gu}{\let\PYG@bf=\textbf\def\PYG@tc##1{\textcolor[rgb]{0.50,0.00,0.50}{##1}}}
\@namedef{PYG@tok@gd}{\def\PYG@tc##1{\textcolor[rgb]{0.63,0.00,0.00}{##1}}}
\@namedef{PYG@tok@gi}{\def\PYG@tc##1{\textcolor[rgb]{0.00,0.52,0.00}{##1}}}
\@namedef{PYG@tok@gr}{\def\PYG@tc##1{\textcolor[rgb]{0.89,0.00,0.00}{##1}}}
\@namedef{PYG@tok@ge}{\let\PYG@it=\textit}
\@namedef{PYG@tok@gs}{\let\PYG@bf=\textbf}
\@namedef{PYG@tok@gp}{\let\PYG@bf=\textbf\def\PYG@tc##1{\textcolor[rgb]{0.00,0.00,0.50}{##1}}}
\@namedef{PYG@tok@go}{\def\PYG@tc##1{\textcolor[rgb]{0.44,0.44,0.44}{##1}}}
\@namedef{PYG@tok@gt}{\def\PYG@tc##1{\textcolor[rgb]{0.00,0.27,0.87}{##1}}}
\@namedef{PYG@tok@err}{\def\PYG@bc##1{{\setlength{\fboxsep}{\string -\fboxrule}\fcolorbox[rgb]{1.00,0.00,0.00}{1,1,1}{\strut ##1}}}}
\@namedef{PYG@tok@kc}{\let\PYG@bf=\textbf\def\PYG@tc##1{\textcolor[rgb]{0.00,0.50,0.00}{##1}}}
\@namedef{PYG@tok@kd}{\let\PYG@bf=\textbf\def\PYG@tc##1{\textcolor[rgb]{0.00,0.50,0.00}{##1}}}
\@namedef{PYG@tok@kn}{\let\PYG@bf=\textbf\def\PYG@tc##1{\textcolor[rgb]{0.00,0.50,0.00}{##1}}}
\@namedef{PYG@tok@kr}{\let\PYG@bf=\textbf\def\PYG@tc##1{\textcolor[rgb]{0.00,0.50,0.00}{##1}}}
\@namedef{PYG@tok@bp}{\def\PYG@tc##1{\textcolor[rgb]{0.00,0.50,0.00}{##1}}}
\@namedef{PYG@tok@fm}{\def\PYG@tc##1{\textcolor[rgb]{0.00,0.00,1.00}{##1}}}
\@namedef{PYG@tok@vc}{\def\PYG@tc##1{\textcolor[rgb]{0.10,0.09,0.49}{##1}}}
\@namedef{PYG@tok@vg}{\def\PYG@tc##1{\textcolor[rgb]{0.10,0.09,0.49}{##1}}}
\@namedef{PYG@tok@vi}{\def\PYG@tc##1{\textcolor[rgb]{0.10,0.09,0.49}{##1}}}
\@namedef{PYG@tok@vm}{\def\PYG@tc##1{\textcolor[rgb]{0.10,0.09,0.49}{##1}}}
\@namedef{PYG@tok@sa}{\def\PYG@tc##1{\textcolor[rgb]{0.73,0.13,0.13}{##1}}}
\@namedef{PYG@tok@sb}{\def\PYG@tc##1{\textcolor[rgb]{0.73,0.13,0.13}{##1}}}
\@namedef{PYG@tok@sc}{\def\PYG@tc##1{\textcolor[rgb]{0.73,0.13,0.13}{##1}}}
\@namedef{PYG@tok@dl}{\def\PYG@tc##1{\textcolor[rgb]{0.73,0.13,0.13}{##1}}}
\@namedef{PYG@tok@s2}{\def\PYG@tc##1{\textcolor[rgb]{0.73,0.13,0.13}{##1}}}
\@namedef{PYG@tok@sh}{\def\PYG@tc##1{\textcolor[rgb]{0.73,0.13,0.13}{##1}}}
\@namedef{PYG@tok@s1}{\def\PYG@tc##1{\textcolor[rgb]{0.73,0.13,0.13}{##1}}}
\@namedef{PYG@tok@mb}{\def\PYG@tc##1{\textcolor[rgb]{0.40,0.40,0.40}{##1}}}
\@namedef{PYG@tok@mf}{\def\PYG@tc##1{\textcolor[rgb]{0.40,0.40,0.40}{##1}}}
\@namedef{PYG@tok@mh}{\def\PYG@tc##1{\textcolor[rgb]{0.40,0.40,0.40}{##1}}}
\@namedef{PYG@tok@mi}{\def\PYG@tc##1{\textcolor[rgb]{0.40,0.40,0.40}{##1}}}
\@namedef{PYG@tok@il}{\def\PYG@tc##1{\textcolor[rgb]{0.40,0.40,0.40}{##1}}}
\@namedef{PYG@tok@mo}{\def\PYG@tc##1{\textcolor[rgb]{0.40,0.40,0.40}{##1}}}
\@namedef{PYG@tok@ch}{\let\PYG@it=\textit\def\PYG@tc##1{\textcolor[rgb]{0.24,0.48,0.48}{##1}}}
\@namedef{PYG@tok@cm}{\let\PYG@it=\textit\def\PYG@tc##1{\textcolor[rgb]{0.24,0.48,0.48}{##1}}}
\@namedef{PYG@tok@cpf}{\let\PYG@it=\textit\def\PYG@tc##1{\textcolor[rgb]{0.24,0.48,0.48}{##1}}}
\@namedef{PYG@tok@c1}{\let\PYG@it=\textit\def\PYG@tc##1{\textcolor[rgb]{0.24,0.48,0.48}{##1}}}
\@namedef{PYG@tok@cs}{\let\PYG@it=\textit\def\PYG@tc##1{\textcolor[rgb]{0.24,0.48,0.48}{##1}}}

\def\PYGZus{\char`\_}

\def\PYGZlt{\char`\<}

\def\PYGZhy{\char`\-}
\def\PYGZsq{\char`\'}
\def\PYGZdq{\char`\"}

\def\hyper@link#1#2#3{#3}%
\let\hyper@linkstart\@gobbletwo
\let\hyper@linkend\@empty
\makeatother

\usepackage[frozencache, cachedir=.]{minted}

\usepackage{authblk}
\title{Splitting vertices of bipartite graphs preserves de Bruijn-Erd\H{o}s property}
\author[1]{Laurent Beaudou\footnote{\tt laurent.beaudou@uca.fr}}
\author[2]{Guillermo Gamboa Quintero\footnote{\tt gamboa@iuuk.mff.cuni.cz}}
\affil[1]{Université Clermont-Auvergne, CNRS, Mines de Saint-Étienne\\
Clermont-Auvergne-INP, LIMOS, 63000 Clermont-Ferrand, France}
\affil[2]{Computer Science Institute of Charles University, Charles University, Prague, Czechia}
\date{\today}                     
\newtheorem{theorem}{Theorem}
\newtheorem{lemma}[theorem]{Lemma}
\newtheorem*{problem}{Problem}
\newtheorem*{conj}{Chen-Chv\'atal conjecture}
\theoremstyle{remark}
\newtheorem{claim}[theorem]{Claim}

\newcommand*{\Line}[1]{\ensuremath{\overline{#1}}}
\newcommand*{\Distt}{dist }
\newcommand*{\Distm}{\text{dist}}
\newcommand*{\Distx}[1]{\ensuremath{\text{dist}(#1)}}

\begin{document}

\maketitle

\begin{abstract}
  In this note, we prove that every graph obtained from a bipartite
  graph by iteratively splitting vertices into two adjacent twins has
  the de~Bruijn-Erd\H{o}s property.
\end{abstract}


Given a metric space $(X,\Distm)$, we follow the natural definition of betweenness introduced by Menger~\cite{m1928} in 1928: an element $b$ is {\em between} elements $a$ and~$c$ if $\Distx{a,b} + \Distx{b,c} = \Distx{a,c}$. We say that elements $a,b$ and $c$ are {\em collinear} if they are pairwise distinct and one of these is between the other two. Eight decades after Menger's work, Chen and Chv\'atal~\cite{cc2008} introduced the notion of line for general metric spaces. Namely, given two points $a$ and $b$ of $X$, the {\em line generated by $a$ and $b$}, denoted \Line{ab}, is the set $\{a,b\}$ augmented by all points forming a collinear triple with $a$ and $b$. In their paper, they formulate a question which has since earned the title of conjecture.

\begin{conj}[Question 1 in \cite{cc2008}]
  Every finite metric space $(X,\Distm)$ where no line consists of the
  entire ground set $X$ determines at least $\lvert X \rvert$ distinct lines.
\end{conj}

Their original goal was to extend some results from the Euclidean setting to the more general setting of metric spaces. In this specific case, the authors refer to the original work of de Bruijn and Erd\H{o}s~\cite{dbe46}. Therefore, we say that a finite metric space $(X,\Distm)$ has the {\em de Bruijn-Erd\H{o}s Property} if it has a {\em universal line} (that is, a line containing all points of $X$) or at least $ \lvert X \rvert$ lines.

The Chen-Chv\'atal conjecture remains open in the restricted case of metric spaces arising from graphs where $X$ is the vertex set and \Distt is the usual shortest path distance. In his last survey published of 2018, Chv\'atal~\cite{chvatal2018} lists no less than twenty-nine open problems related to the Chen-Chv\'{a}tal conjecture. Some of these have been solved since. We focus ourselves on the eighth problem of the survey.

\begin{problem}[Problem 8 in \cite{chvatal2018}]
Prove that all graphs obtained from bipartite graphs by repeatedly splitting of vertices into adjacent twins have the de Bruijn-Erd\H{o}s property.
\end{problem}

\noindent In this short note, we answer positively to this problem via Theorem~\ref{thm:main}.


\begin{theorem}
  \label{thm:main}
  Let $G$ be a graph on $n$ vertices (with $n \geq 2$) obtained from a
  bipartite graph by repeated splitting of vertices into adjacent
  twins, then $G$ admits a universal line or has at least $n$ distinct
  lines.
\end{theorem}


In order to prove Theorem~\ref{thm:main} we use a few known results
that we recall here.

\begin{lemma}
  \label{lem:disconnected}
  If $G$ is a disconnected graph on two vertices or more, then $G$
  has the de Bruijn-Erd\H{o}s property.
\end{lemma}

\begin{proof}
  Let $G$ be a disconnected graph on $n$ vertices, then there exists a
  non trivial partition of its vertex set into two sets $V_1$ of $n_1$
  vertices and $V_2$ of $n_2$ vertices such that no edge goes between
  $V_1$ and $V_2$. Therefore any pair of vertices $(u,v)$ in
  $V_1\times V_2$ generates a line of cardinality 2: $\Line{uv} =
  \{u,v\}$. We thus have $n_1n_2$ different lines. If both $n_1$ and
  $n_2$ are greater than or equal to 2, $n_1n_2 \geq n$. If exactly
  one of them (say $n_1$) is 1, then we get $n_2 = n-1$, another line
  is obtained by taking any pair of vertices in $V_2$. If both of them
  are 1, then the graph has two vertices and the line they define is
  universal.
\end{proof}

\begin{lemma}[Theorem 6 in~\cite{chvatal2018}]
  \label{lem:bipuniv}
  If $G$ is a bipartite connected graph, then any pair of adjacent
  vertices generates a universal line.
\end{lemma}

\begin{proof}
  In a connected graph $G$, the line defined by a pair of adjacent vertices contains every vertex except those equidistant to both vertices. Since $G$ is connected, this enforces the existence of an odd cycle.
\end{proof}

\section*{Proof of Theorem~\ref{thm:main}}

Let $G$ be a graph on $n$ vertices obtained from a bipartite graph $H$ by repeated splitting of vertices into adjacent twins. In other words, there is a mapping $f : V(G) \rightarrow V(H)$ which assigns to each vertex of $G$ the original vertex in $H$ from which it has been derived. Note that this mapping is surjective (we only split and never delete so every vertex of $H$ has at least one vertex in $G$ being mapped to it). For a given vertex $x$ of $H$, the {\em blob} of $x$ is its preimage by $f$, denoted by $f^{-1}(x)$. Observe that a blob forms a clique in $G$. A blob is {\em trivial} if it is a singleton, and {\em rich} otherwise. Figure~\ref{fig:blobs} depicts a transformation from a bipartite graph $H$ to a graph $G$ with five rich blobs and two trivial blobs.


\begin{figure}[ht]
    \begin{centering}
      
      \begin{tikzpicture}[
          vertex/.style={circle, draw, fill=black, minimum size=2pt, inner sep=0pt},
          scale = .6
        ]

        \pgfmathsetmacro{\r}{.2}
        \pgfmathsetmacro{\t}{.6}

        \node[vertex] (u1) at (0, 2) {};
        \node[vertex] (u2) at (2, 2) {};
        \node[vertex] (u3) at (4, 2) {};
        
        \node[vertex] (v1) at (0, 0) {};
        \node[vertex] (v2) at (2, 0) {};
        \node[vertex] (v3) at (4, 0) {};

        \node[vertex] (z) at (3,1) {};
        \draw (z) -- (u2);
        \draw (z) -- (v3);
        
        \draw (u1) -- (v1);
        \draw (u1) -- (u2);
        \draw (u2) -- (v2);
        \draw (u2) -- (u3);
        \draw (u3) -- (v3);
        \draw (v1) -- (v2);
        \draw (v2) -- (v3);

        \begin{scope}[xshift=6cm, yshift=4cm]
        
        \node[vertex] (u1) at (0, 2) {};
        \node[vertex] (u2) at (2, 2) {};
        \node[vertex] (u3) at (4, 2) {};
        
        \node[vertex] (v1) at (0, 0) {};
        \node[vertex] (v21) at (1.9, 0.1) {};
        \node[vertex] (v22) at (2.1,-0.1) {};
        \node[vertex] (v3) at (4, 0) {};

        \node[vertex] (z) at (3,1) {};
        \draw (z) -- (u2);
        \draw (z) -- (v3);
        
        \draw (u1) -- (v1);
        \draw (u1) -- (u2);
        \draw (u2) -- (v21);
        \draw (u2) -- (v22);
        \draw (u2) -- (u3);
        \draw (u3) -- (v3);
        \draw (v1) -- (v21);
        \draw (v22) -- (v3);
        \draw (v1) -- (v22);
        \draw (v21) -- (v3);
        \draw (v21) -- (v22);

        \draw[dotted] (2,0) circle (.3);

        \end{scope}
        \begin{scope}[xshift = 12cm]

        \node[vertex] (u11) at ($(195:\t) + (0, 2)$) {};
        \node[vertex] (u12) at ($(75:\t) + (0, 2)$) {};
        \node[vertex] (u13) at ($(315:\t) + (0, 2)$) {};
        \node[vertex] (u21) at ($(45:\r) + (2, 2)$) {};
        \node[vertex] (u22) at ($(-135:\r) + (2, 2)$) {};
        \node[vertex] (u3) at (4, 2) {};
        
        \node[vertex] (v1) at (0, 0) {};
        \node[vertex] (v21) at ($(45:\r) + (2, 0)$) {};
        \node[vertex] (v22) at ($(-135:\r) + (2, 0)$) {};
        \node[vertex] (v31) at ($(15:\t) + (4, 0)$) {};
        \node[vertex] (v32) at ($(135:\t) + (4, 0)$) {};
        \node[vertex] (v33) at ($(255:\t) + (4, 0)$) {};

        \node[vertex] (z1) at ($(45:\r) + (3, 1)$) {};
        \node[vertex] (z2) at ($(-135:\r) + (3, 1)$) {};
        \draw (z1) -- (v31);
        \draw (z1) -- (v32);
        \draw (z1) -- (v33);
        \draw (z1) -- (u21);
        \draw (z1) -- (u22);
        \draw (z2) -- (v31);
        \draw (z2) -- (v32);
        \draw (z2) -- (v33);
        \draw (z2) -- (u21);
        \draw (z2) -- (u22);
        \draw (z1) -- (z2);

        \draw (u11) -- (u12) -- (u13) -- (u11);
        \draw (u21) -- (u22);
        \draw (v21) -- (v22);
        \draw (v31) -- (v32) -- (v33) -- (v31);

        \draw (u21) -- (v21);
        \draw (u21) -- (v22);
        \draw (u22) -- (v21);
        \draw (u22) -- (v22);

        \draw (u11) -- (u21) -- (u12) -- (u22) -- (u13) -- (u21);
        \draw (u11) -- (u22);

        \draw (u21) -- (u3) -- (u22);

        \draw (v31) -- (v21) -- (v32) -- (v22) -- (v33) -- (v21);
        \draw (v31) -- (v22);

        \draw (v21) -- (v1) -- (v22);

        \draw (u11) -- (v1) -- (u12);
        \draw (u13) -- (v1);

        \draw (v31) -- (u3) -- (v32);
        \draw (v33) -- (u3);

        \draw[dotted] (2,0) circle (\r + .2);
        \draw[dotted] (3,1) circle (\r + .2);
        \draw[dotted] (4,0) circle (\t + .2);
        \draw[dotted] (0,2) circle (\t + .2);
        \draw[dotted] (2,2) circle (\r + .2);
        
        \end{scope}

        \node at (1.7,4.8) {\scriptsize one split};
        \node at (14.8,4.8) {\scriptsize several splits};
        \node at (2,-.5) {\scriptsize $H$};
        \node at (14,-.8) {\scriptsize $G$};
        
        \draw [->] (2,3) to [out=90, in=180] (5,5);
        \draw [->, dashed] (11,5) to [out=0, in=90] (14,3);
        
      \end{tikzpicture}

    \end{centering}
    \caption{Splitting vertices from some graph $H$ to some graph $G$.}
    \label{fig:blobs}
  \end{figure}
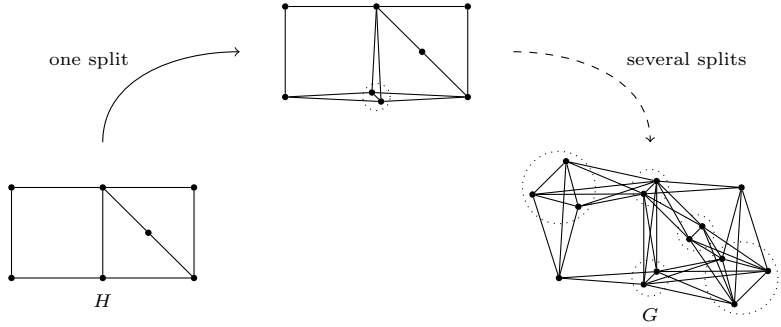


If $G$ is disconnected, then by Lemma~\ref{lem:disconnected}, we are done. So, we now assume that $G$ is connected.

\begin{claim}
  \label{claim:connected}
  We may assume that $G$ is connected.
\end{claim}

Let us now focus on the line generated by any two vertices $u$ and $v$ of $G$. 

\begin{itemize}
    \item If $u$ and $v$ are in the same blob, then note that $\Line{uv} = \{u,v\}$, so it is uniquely generated.
    \item If $u$ and $v$ are in different blobs (meaning $f(u) \neq f(v)$), consider the line \Line{f(u)f(v)} in $H$. Then, the line \Line{uv} consists of $u$, $v$ and all points $z$ such that $f(z) \in \Line{f(u)f(v)}$ and $f(z) \notin \{ f(u), f(v)\}$.
\end{itemize}

The second case allows us to conclude the following: If two trivial
blobs $u$ and $v$ are adjacent, then by Lemma~\ref{lem:bipuniv} and
the previous discussion, line \Line{uv} is universal. So, we may further assume the following:

\begin{claim}
  \label{claim:trivialstable}
  We may assume that trivial blobs form an independent set.
\end{claim}

As a direct corollary of Claims~\ref{claim:connected}
and~\ref{claim:trivialstable}, we can further our assumptions to the following:

\begin{claim}
  \label{claim:trivialconnected}
  Every trivial blob has a rich blob neighbor.
\end{claim}

We may now define a pair of parameters that will allow us to lower bound
the number of lines in $G$:
\begin{itemize}
    \item let $p$ be the number of vertices in $G$ belonging to a rich blob (in Figure~\ref{fig:blobs}, $p=12$),
    \item let $k$ be the number of rich blobs having at least one trivial   blob as a neighbor (in Figure~\ref{fig:blobs}, $k=4$).
\end{itemize}

Observe that since each rich blob has at least two vertices and there
are at least $k$ reach blobs, then

\begin{equation}
  \label{eq:pk}
  p \geq 2k
\end{equation}

Any line generated by a pair of vertices among rich blobs is uniquely
defined (if both $u$ and $v$ are in a common blob, the line has
cardinality 2, if in different blobs, the line intersects exactly two
blobs partially isolating $u$ and $v$). Thus, we get $\binom{p}{2}$
distinct lines. Therefore if $p = n$, we obtain $\binom{n}{2}$ lines
and $G$ has the de~Bruijn-Erd\H{o}s property. Moreover, if no
splitting has been made, then $p = 0$ and $G$ is bipartite and we
know that has the de~Bruijn-Erd\H{o}s property. So, we may assume the following:

\begin{claim}
  We may assume that $2 \leq p \leq n-1$.
\end{claim}

Whenever $p < n$, there are some trivial blobs (exactly $n-p$). Thus
by Claim~\ref{claim:trivialconnected}, $k$ cannot be 0. Then, there
exists a rich blob $B$ which has at least $\lceil \frac{n-p}{k}\rceil$
trivial blobs as neighbors. Consider the set of lines generated by
pairs of trivial blobs adjacent to $B$. These trivial blobs are all at
distance 2 one from each other and so any such lines takes only two of
them. We get a set of $\binom{\lceil \frac{n-p}{k} \rceil}{2}$ different
lines. These lines either take a blob completely or do not intersect
it. Thus they are distinct from the $\binom{p}{2}$ lines previously
defined.

Finally, for each vertex $u$ in one of the $k$ blobs having a trivial
neighbor, we consider a line \Line{uv} where $v$ is one of the
trivial blobs adjacent to $u$. These lines intersect exactly one blob
partially isolating $u$. Therefore they are distinct from each other
and new with respect to both previous families of lines. Since every rich blob
has at least two vertices, they form a family of at least $2k$
lines. In conclusion, the number of lines in $G$ is at least,
\begin{equation}
  \label{eq:base}
  \binom{p}{2} + \binom{\lceil \frac{n-p}{k} \rceil}{2} + 2k.
\end{equation}

The end of the proof amounts to proving that this quantity is greater
than or equal to $n$ for any choice of $p$ and $k$ such that $2 \leq p
\leq n-1$ and $1 \leq k \leq p/2$. It turns out that the additive term
of $2k$ is needed only for very specific cases when $n$ is
small. Cases for low values (namely when $n \leq 39$) are thus
examined through a basic computer program (see
Appendix~\ref{app:programme}). In general the first two terms
of~\eqref{eq:base} are enough.

Let $\epsilon = 1.531$ and $n_0 = 40$. These numbers have been chosen
in order to minimize $n_0$ while satisfying the following implications
for any $n \geq n_0$ and any real number $x \geq 1$. The proofs are
not difficult but require paper space. So we present them in
Appendix~\ref{app:inequalities}.
\begin{align}
  x \geq \epsilon \sqrt{n} &\Rightarrow \frac{x^2-x}{2} \geq n, \label{eq:1}\\
  x \geq \epsilon \sqrt{\frac{n}{2}} &\Rightarrow \frac{x^2 -x }{2} \geq \frac{n}{2}, \label{eq:2}\\
  x \leq \epsilon \sqrt{n} &\Rightarrow \frac{2n}{x} - 2 \geq \epsilon \sqrt{\frac{n}{2}}, \label{eq:3}\\
  \frac{2 \sqrt{n}}{\epsilon + \frac{2}{\sqrt{n}}}  & \geq \epsilon \sqrt{\frac{n}{2}}. \label{eq:4}
\end{align}
The basic sketch is the following: if $p$ is sufficiently large (greater than
$\epsilon \sqrt{n}$), then the first term is greater than $n$. If $p$
is sufficiently small (less than $\frac{2\sqrt{n}}{\epsilon + 2/\sqrt{n}}$), then the
second term is greater than $n$. If $p$ is in between, each term is at least
$n/2$ summing up to $n$.

\paragraph{Case 1: $p \geq \epsilon\sqrt{n}$.}
When $p$ is large enough we consider only the first term $\binom{p}{2}$. By~\eqref{eq:1}, 
\begin{align*}
  \binom{p}{2} & = \frac{p^2 - p}{2}\\
  & \geq n.
\end{align*}

\paragraph{Case 2: $p \leq \frac{2 \sqrt{n}}{\epsilon + \frac{2}{\sqrt{n}}}$.}
When $p$ is small enough, we consider only the second term. First we observe that $k$ being at most $p/2$, we have
$$\left\lceil \frac{n-p}{k} \right\rceil \geq \frac{n-p}{k} \geq \frac{2n}{p}-2.$$
By the upper bound on $p$, we obtain:
\begin{equation*}
  \frac{2n}{p}-2  \geq \epsilon \sqrt{n};
\end{equation*}
and using~\eqref{eq:1}, we conclude
\begin{equation*}
  \binom{\left\lceil \frac{n-p}{k} \right\rceil}{2} \geq \binom{\frac{2n}{p} - 2}{2} \geq n.
\end{equation*}

\paragraph{Case 3: $ \frac{2 \sqrt{n}}{\epsilon + \frac{2}{\sqrt{n}}} < p < \epsilon \sqrt{n}$.}
In this last case, we prove that both first terms of the
sum~\eqref{eq:base} provide at least $n/2$ lines. Observe that by~\eqref{eq:3} and~\eqref{eq:4},
\begin{align*}
  p < \epsilon \sqrt{n} & \text{ thus } \frac{2n}{p} - 2 \geq \epsilon \sqrt{\frac{n}{2}}, \\
  p > \frac{2 \sqrt{n}}{\epsilon + \frac{2}{\sqrt{n}}} & \text{ thus } p \geq \epsilon \sqrt{\frac{n}{2}}.
\end{align*}
Using~\eqref{eq:2} completes the proof.

\subsection*{Acknowledgment}
This research was done while G. Gamboa Quintero visited L. Beaudou in Clermont-Ferrand during Spring 2025 with the support of ERASMUS+. The author G. Gamboa Quintero was supported by Charles University project 
PRIMUS/24/SCI/012.

\bibliographystyle{alpha}
\bibliography{sample}

\newpage

\appendix

\section{Python code checking triples up to $n = 39$}
\label{app:programme}

This script can be downloaded from the GitLab\footnote{\href{https://gitlab.limos.fr/labeaudo/check_triples_lines}{\tt https://gitlab.limos.fr/labeaudo/check\_triples\_lines}} of first author.

\begin{center}
  \begin{minipage}{.9\textwidth}

    \begin{Verbatim}[commandchars=\\\{\}]
\PYG{k+kn}{from} \PYG{n+nn}{numpy} \PYG{k+kn}{import} \PYG{n}{ceil}\PYG{p}{,} \PYG{n}{floor}
\PYG{k+kn}{from} \PYG{n+nn}{math} \PYG{k+kn}{import} \PYG{n}{comb}

\PYG{k}{def} \PYG{n+nf}{nbLines}\PYG{p}{(}\PYG{n}{n}\PYG{p}{,}\PYG{n}{p}\PYG{p}{,}\PYG{n}{k}\PYG{p}{):}
\PYG{+w}{    }\PYG{l+s+sd}{\PYGZsq{}\PYGZsq{}\PYGZsq{} Returns a lower bound on the number of lines in a graph G}
\PYG{l+s+sd}{    obtained by iteratively splitting the vertices of a bipartite}
\PYG{l+s+sd}{    graph H into true twins.}

\PYG{l+s+sd}{    Input: n, number of vertices of G}
\PYG{l+s+sd}{           p, number of vertices of G in rich blobs}
\PYG{l+s+sd}{           k, number of rich blobs with trivial neighbours}

\PYG{l+s+sd}{    Output: an integer lower bounding the number of lines in G.}
\PYG{l+s+sd}{    \PYGZsq{}\PYGZsq{}\PYGZsq{}}
    \PYG{k}{assert} \PYG{l+m+mi}{2} \PYG{o}{\PYGZlt{}=} \PYG{n}{p} \PYG{o+ow}{and} \PYG{n}{p} \PYG{o}{\PYGZlt{}=} \PYG{n}{n}\PYG{o}{\PYGZhy{}}\PYG{l+m+mi}{1} \PYG{o+ow}{and} \PYG{l+m+mi}{1}\PYG{o}{\PYGZlt{}=} \PYG{n}{k} \PYG{o+ow}{and} \PYG{n}{k} \PYG{o}{\PYGZlt{}=} \PYG{n}{p}\PYG{o}{/}\PYG{l+m+mi}{2}
    \PYG{n}{first} \PYG{o}{=} \PYG{n}{comb}\PYG{p}{(}\PYG{n}{p}\PYG{p}{,}\PYG{l+m+mi}{2}\PYG{p}{)}
    \PYG{n}{second} \PYG{o}{=} \PYG{n}{comb}\PYG{p}{((}\PYG{n+nb}{int}\PYG{p}{(}\PYG{n}{ceil}\PYG{p}{((}\PYG{n}{n}\PYG{o}{\PYGZhy{}}\PYG{n}{p}\PYG{p}{)}\PYG{o}{/}\PYG{n}{k}\PYG{p}{))),}\PYG{l+m+mi}{2}\PYG{p}{)}
    \PYG{n}{third} \PYG{o}{=} \PYG{l+m+mi}{2}\PYG{o}{*}\PYG{n}{k}
    \PYG{k}{return} \PYG{n}{first} \PYG{o}{+} \PYG{n}{second} \PYG{o}{+} \PYG{n}{third}

\PYG{k}{def} \PYG{n+nf}{main}\PYG{p}{(}\PYG{n}{n0}\PYG{p}{):}
\PYG{+w}{    }\PYG{l+s+sd}{\PYGZsq{}\PYGZsq{}\PYGZsq{} Main routine, enumerates all legal triples (n,p,k) for n}
\PYG{l+s+sd}{    between 3 and n0. And checks that the number of lines is at least}
\PYG{l+s+sd}{    n. If some triple fails, it raises an Exception.}
\PYG{l+s+sd}{    \PYGZsq{}\PYGZsq{}\PYGZsq{}}
    \PYG{k}{for} \PYG{n}{n} \PYG{o+ow}{in} \PYG{n+nb}{range}\PYG{p}{(}\PYG{l+m+mi}{3}\PYG{p}{,}\PYG{n}{n0}\PYG{o}{+}\PYG{l+m+mi}{1}\PYG{p}{):}
       \PYG{k}{for} \PYG{n}{p} \PYG{o+ow}{in} \PYG{n+nb}{range}\PYG{p}{(}\PYG{l+m+mi}{2}\PYG{p}{,}\PYG{n}{n}\PYG{p}{):}
            \PYG{k}{for} \PYG{n}{k} \PYG{o+ow}{in} \PYG{n+nb}{range}\PYG{p}{(}\PYG{l+m+mi}{1}\PYG{p}{,}\PYG{n+nb}{int}\PYG{p}{(}\PYG{n}{floor}\PYG{p}{(}\PYG{n}{p}\PYG{o}{/}\PYG{l+m+mi}{2}\PYG{p}{))}\PYG{o}{+}\PYG{l+m+mi}{1}\PYG{p}{):}
                \PYG{k}{if} \PYG{n}{nbLines}\PYG{p}{(}\PYG{n}{n}\PYG{p}{,}\PYG{n}{p}\PYG{p}{,}\PYG{n}{k}\PYG{p}{)} \PYG{o}{\PYGZlt{}} \PYG{n}{n}\PYG{p}{:}
                    \PYG{k}{raise} \PYG{n+ne}{Exception}\PYG{p}{(}\PYG{l+s+s2}{\PYGZdq{}We have a problem.\PYGZdq{}}\PYG{p}{)}
    \PYG{n+nb}{print}\PYG{p}{(}\PYG{l+s+s2}{\PYGZdq{}Checked.\PYGZdq{}}\PYG{p}{)}

\PYG{k}{if} \PYG{n+nv+vm}{\PYGZus{}\PYGZus{}name\PYGZus{}\PYGZus{}}\PYG{o}{==}\PYG{l+s+s2}{\PYGZdq{}\PYGZus{}\PYGZus{}main\PYGZus{}\PYGZus{}\PYGZdq{}}\PYG{p}{:}
    \PYG{n}{main}\PYG{p}{(}\PYG{l+m+mi}{39}\PYG{p}{)}

\end{Verbatim}

    \end{minipage}
  \end{center}

\newpage

\section{Proofs for $\epsilon$ an $n_0$}
\label{app:inequalities}

We want to prove that for $\epsilon = 1.531$ and $n_0 = 40$. The
following is true for any $n \geq n_0$ and any real number $x \geq 1$.
\begin{align}
  x \geq \epsilon \sqrt{n} &\Rightarrow \frac{x^2-x}{2} \geq n, \tag{\ref{eq:1}}\\
  x \geq \epsilon \sqrt{\frac{n}{2}} &\Rightarrow \frac{x^2 -x }{2} \geq \frac{n}{2}, \tag{\ref{eq:2}}\\
  x \leq \epsilon \sqrt{n} &\Rightarrow \frac{2n}{x} - 2 \geq \epsilon \sqrt{\frac{n}{2}},\tag{\ref{eq:3}}\\
  \frac{2 \sqrt{n}}{\epsilon + \frac{2}{\sqrt{n}}}  & \geq \epsilon \sqrt{\frac{n}{2}}. \tag{\ref{eq:4}}
\end{align}

First, observe that \eqref{eq:2} implies \eqref{eq:1} since for any $n
\geq n_0$, we can let $n'$ be $2n$ which is also greater than $n_0$
and then \eqref{eq:2} on $n'$ reads exactly as \eqref{eq:1} on
$n$. Now let us prove \eqref{eq:2}. Since $\epsilon \sqrt{\frac{n}{2}}
> 1$ by our choice of $n_0$ and $\epsilon$, and as function $f: x
\mapsto x^2 - x$ is monotone and increasing on $ \left [ 1/2, + \infty \right [$, we
have:
\begin{align*}
  \frac{x^2-x}{2} &\geq \frac{\epsilon^2 \frac{n}{2} - \epsilon \sqrt{\frac{n}{2}}}{2}\\
  &\geq \epsilon^2 \frac{n}{4} - \epsilon \sqrt{\frac{n}{8}}\\
  &\geq \frac{n}{2} + (\epsilon^2 - 2) \frac{n}{4} - \epsilon \sqrt{\frac{n}{8}}\\
  &\geq \frac{n}{2} + \frac{\sqrt{n}}{2} \left ( (\epsilon^2 -2)\frac{\sqrt{n}}{2} - \frac{\epsilon}{\sqrt{2}} \right).
\end{align*}
The last multiplicative term is increasing with $n$ since $\epsilon >
\sqrt{2}$ when $n = 40$ it is strictly positive ($\simeq 0.005$) thus
validating implications~\eqref{eq:1} and~\eqref{eq:2}.

Concerning implication~\eqref{eq:3},
\begin{align*}
  \frac{2n}{x} - 2 & \geq \frac{2n}{\epsilon\sqrt{n}} - 2\\
  & \geq \frac{2 \sqrt{2}}{\epsilon} \sqrt{\frac{n}{2}} - 2\\
  & \geq \epsilon \sqrt{\frac{n}{2}} + \left ( \frac{2 \sqrt{2}}{\epsilon} - 1 \right )\sqrt{\frac{n}{2}} - 2
\end{align*}
Since $\epsilon < 2 \sqrt{2}$, the additive term is increasing with $n$. It happens to be positive for $n_0$ ($\simeq 0.68$), proving implication~\eqref{eq:3}.

Finally, we check the last inequality~\eqref{eq:4}:
\begin{align*}
  \frac{2 \sqrt{n}}{\epsilon + \frac{2}{\sqrt{n}}}  \geq \epsilon \sqrt{\frac{n}{2}} & \Leftrightarrow 2 \sqrt{n} \geq \epsilon^2 \sqrt{\frac{n}{2}} + \epsilon \sqrt{2} \\
  & \Leftrightarrow \sqrt{n} \geq \frac{\epsilon \sqrt{2}}{2 - \epsilon^2/\sqrt{2}}
\end{align*}
The right hand side term is strictly less than 6.321 while the square
root of 40 is strictly greater than 6.324.


\end{document}